\documentclass{llncs}

\usepackage{amsmath,amsthm}
\usepackage{amssymb}
\usepackage{comment}

\begin{document}

\title{Trading off Worst and Expected Cost in Decision Tree Problems and a Value Dependent Model}
\author{Aline Saettler\inst{1} \and Eduardo Laber\inst{1}  \and Ferdinando Cicalese\inst{2}}
\institute{PUC-Rio, Brazil\\\email{\{asaettler,laber\}@inf.puc-rio.br} \and University of Salerno, Italy\\\email{cicalese@dia.unisa.it}}
\maketitle

\begin{abstract}
We study the problem of evaluating a discrete function by adaptively querying the values of its variables until the values read uniquely determine the value of the function. Reading the value of a variable is done at the expense of some cost, and the goal is to design a strategy (decision tree) for evaluating the function incurring as little cost as possible in the worst case or in expectation (according to a prior distribution on the possible variables assignments). Except for particular cases of the problem, in general, only the minimization of one of these two measures is addressed in the literature. However, there are instances of the problem for which the minimization of one measure leads to a strategy with a high cost with respect to the other measure (even exponentially bigger than the optimal).  We provide a new construction which can guarantee a trade-off between the two criteria. More precisely, given a decision tree guaranteeing expected cost $E$ and a decision tree guaranteeing worst cost $W$ our method can guarantee for any chosen trade-off value $\rho$ to produce a decision tree whose worst cost is $(1 + \rho)W$ and whose expected cost is
$(1 + \frac{1}{\rho})E.$ These bounds are improved for the  relevant case of uniform testing costs.

Motivated by applications, we also study a variant of the problem where the cost of reading a variable depends on the variable’s value.
We provide an $O(\log n)$ approximation algorithm for the minimization of the worst cost measure, which is best possible
under the assumption $P \neq NP.$.
\end{abstract}

\newcommand{\remove}[1]{}

\section{Introduction}

Decision tree construction is a central problem in several areas of computer science, e.g., in data base theory, in computational learning and in artificial intelligence in general. 
In a typical scenario there are several possible hypotheses, which can explain some unknown phenomenon and we want to decide which hypothesis provides the 
correct explanation. We have a prior distribution on the hypotheses and we can use tests to
discriminate among the hypotheses. Each test's outcome eliminates some of the hypotheses, and the set of tests 
is complete, in the sense that by using all the tests we can definitely find the correct hypothesis. Moreover, different tests may have different associated costs. 
The aim is to define the best testing strategy that allows to reach the correct decision spending as little as possible	. 
If the testing is adaptive a strategy is representable by a tree (called decision tree) with each node being a test and each leaf being a hypothesis. 
In a generalization of this scenario, one is only interested in identifying a class of possible hypothesis explaining the situation. 

In an example of automatic diagnosis, the hypotheses
 are possible diseases and we look for the testing strategy (decision tree) which can always identify the disease by using a cheap sequence of tests. 
In the case we are interested in deciding the drug to administer to the patient rather than exactly identifying the disease we have an instance of the more general variant 
of the decision tree construction where we are looking for the class of hypotheses containing the correct explanation. 

What is the right measure to optimize when constructing the  decision tree? Usually, the expected cost of the tests needed to reach the correct decision and
the maximum total cost needed to reach the correct decision are used. However, these measures can lead to very different trees and in 
particular it is possible that the decision tree minimizing one measure is very inefficient with respect to the other measure. A very skewed distribution can induce
 a tree optimizing the expected cost with a very skewed shape. As a consequence, in such a tree some decision might induce a very high cost, even exponentially bigger than 
 the worst cost spent by a strategy that optimizes with respect to the worst case. Conversely optimizing with respect to the worst case can lead to very 
 bad expected cost. The choice of which measure to choose is crucial especially since in practical applications the real distribution might not be known but only 
 estimated and possibly be wrong. Therefore, it might  be preferable to have decision trees which while optimizing one criteria guarantees to be 
 efficient with respect to the other. 
 
In this paper, we address the issue regarding the existence of a trade-off between the minimization 
of the worst testing cost and the
expected testing cost of decision trees. Is it is possible to 
construct decision trees that are efficient with 
respect to both measures? As mentioned before, these two goals can be incompatible.

\remove{
We present a polynomial time procedure
 that given a parameter $\rho >0$ and  two decision trees
$D_W$ and $D_E$, the former  with worst testing cost $W$ and the latter with expected testing cost $E$,  produces
a decision tree $D$ with worst testing cost at most $(1+\rho)W$ and expected
testing cost at most $(1+1/\rho)E$. In words, the combined decision tree loses at most small constant factor with respect to the 
performance of both the original decision trees which were optimizing with respect to only one measure.  
For the case of uniform testing costs,
the bound can be improved to $(1+\rho)W$ and $(1+2/(\rho^2 +2 \rho))E$ through a more involved analysis. 
}
 
The second issue on which we focus in this paper is the way the cost of the tests is defined. 
We refer the interested reader to \cite{Turney} and references quoted therein for a 
remarkable account of several types of costs to be taken into account in inference procedures. 
In most decision tree problems, the assumption is that the cost of the tests is fixed in advance and known to the algorithm. In particular, 
the cost is independent of the outcome of the test. However, there are also several scenarios in medical 
applications---one of the main fields motivating automatic diagnosis---where the assumption that a test has a fixed cost independent 
of the outcome of the test does not apply. Many diagnostic tests  actually consist of a multi-stage procedure, e.g.,  
in a first stage the sample is tested against some reagent to check for the presence or absence of an antigene. If this appears to be present below a certain level 
the test is considered to be negative and no further analysis is performed. Otherwise, the test is {\em necessarily} followed by a second stage 
where several new reagents are used with  significantly higher final costs. Notice that in such a situation there is no real decision left to the strategy between the first and the second stage, so it is reasonable to consider such a two stage procedure as a single test whose cost depends on the outcome.

Value dependent test costs are also useful in application where disruptive tests are used.
 Consider the use of bacterial colonies or {\em caviae} to test for toxicity of a samples. In the case no toxicity is found, the testing colony can be reused, as opposed to the case where toxicity is verified leading to the disruption of the colony or the death of the 
{\em cavia} (a similar model has been studied in \cite{Elser-Kleber}). Analogously, a chemical reagent might be used for performing a test and the outcome of the test is either some chemical reaction changing the nature of the reagents and making them unusable again, 
or the absence of the reaction in which case the reagent can be (partially) reused. 
Again we have a test that when positive has higher cost---the necessity of buying new reagents---than in the case of a negative outcome.


For this extended version, where the cost of a test may depend on its outcome, 
we present an algorithm for building a decision tree that aims to  minimize the worst testing cost for 
identifying the class of the correct hypothesis.


\subsection{Problem Formalization} 
\noindent
{\bf The Discrete Function Evaluation Problem} (DFEP). Our results are presented in terms of the problem of evaluating 
a discrete function. This problems generalizes most decision tree construction problems studied in the literature.

%
%

An instance of the problem is defined by  a quintuple  $(S, C, T, {\bf p}, {\bf c}),$ 
where $S = \{s_1, \dots, s_n\}$ is a set of objects, 
$C = \{C_1, \dots, C_m\}$ is a partition of $S$ into $m$ classes, $T$ is a set of tests, ${\bf p}$ is a probability distribution on 
$S,$ and ${\bf c}$ is a cost function assigning to each test $t$ a cost $c(t) \in \mathbb{Q^+}$.

A test $t \in T$,  when applied to an object  $s \in S$, outputs a number $t(s)$ in the set  $\{1,\ldots,\ell\}$ and incurs a cost $c(t)$.
It is assumed that the set of tests is complete, in the sense that for any distinct $s_1, s_2 \in S$ there exists a test
$t$ such that $t(s_1) \neq t(s_2).$ 
The goal is to define a testing procedure which uses tests from $T$ and minimizes the testing cost 
(in expectation and/or in the worst case)
for identifying the class of an unknown object $s^*$ chosen according to the distribution ${\bf p}.$ 
We also work with the extended version of the DFEP where the cost of a test
is a function that assigns each pair (test $t$, object $s$) to a value $c^{t(s)}(t) \in \mathbb{Q^+}$.

The DFEP can be rephrased in terms of minimizing the cost of evaluating a discrete function that 
maps points (corresponding to objects) from some finite subset of $\{1,\ldots,\ell\}^{|T|}$ into
values from $\{1, \dots, m\}$ (corresponding to classes), where  each object $s\in S$ corresponds to the point $(t_1(s),\ldots,t_{|T|}(s))$ 
 which is obtained by applying each test of  $T$ to $s$.  
This perspective motivates the name we chose for the problem. However, for the sake of uniformity with more recent work \cite{golovin,bellala} we employ the definition of the problem in terms of objects/tests/classes.

\medskip

\noindent
{\bf Decision Tree Optimization.} Any testing procedure can be represented by a \emph{decision tree}, which is a tree
where every internal node is associated with a test and every leaf is associated with a set
of objects that belong to the same class.
More formally, a decision tree $D$ for $(S,C,T,\mathbf{p}, {\bf c})$  is a leaf
associated with class $i$ 
if every object of $S$ belongs to the same class $i$. Otherwise, the root $r$  of $D$ is associated with some test $t \in  T$ and  
the children of $r$ are  decision trees for the non empty sets in $\{S_{t}^1,\ldots,S_{t}^\ell\}$,
where $S_{t}^i$ is the subset of $S$ that outputs $i$ for test $t$.
 
Given a decision tree $D$, rooted at $r$, we can identify the class of an unknown object 
$s^*$ by following a path from  $r$  to a leaf as follows:
first, we ask for the result of the test associated with $r$ when performed on $s^*$; then,  
we follow the branch of $r$ associated with the result of the test to
reach a child $r_i$ of $r$; next,  we apply the same steps recursively for the decision tree rooted at $r_i$.
The procedure ends when a leaf is reached, which determines the class of $s^*$.

We define $cost(D,s)$ as the sum of the tests' cost on 
the root-to-leaf path from the root of 
$D$ to the leaf associated with object $s$. 
Then, the \emph{worst testing cost}  and the \emph{expected testing cost} of $D$ are, respectively, defined as

\begin{equation}
cost_W(D) = \max_{s \in S}\{cost(D,s)\} \,\,\, \mbox{ and}  \,\,\, cost_E(D) = \sum_{s \in S} cost(D,s) p(s)
\end{equation}

\subsection{Our Results}
 We present a polynomial time procedure
 that given a parameter $\rho >0$ and  two decision trees
$D_W$ and $D_E$, the former  with worst testing cost $W$ and the latter with expected testing cost $E$,  produces
a decision tree $D$ with worst testing cost at most $(1+\rho)W$ and expected
testing cost at most $(1+1/\rho)E$. For the relevant case of uniform costs,
the bound can be improved to $(1+\rho)W$ and $(1+2/(\rho^2 +2 \rho))E$ through a more involved analysis. 

In addition, we 
present an algorithm 
for the minimization of  the worst testing cost for 
the extended version of the $DFEP$ where the cost of a test  depend on its outcome. 
We prove that our  algorithm is  an $O(ln(n))$ approximation for the 
case of binary tests.  This bound 
 is the best possible
under the assumption that $\mathcal{P} \neq \mathcal{NP}$.

\subsection{Related work}

In a recent paper \cite{labericml}, the authors
show  that for any instance $I$ of the
DFEP, with $n$ objects, it is possible to construct in polynomial time a decision tree
$D$ such that $cost_E(D) $ is $ O(\log n \cdot OPT_E(I))$
and $cost_W(D)$ is $ O(\log n \cdot OPT_W(I))$,
where $OPT_E(I)$ and $OPT_W(I)$ are,
respectively, the minimum expected  testing cost
and the minimum worst testing cost for instance $I$. 

Note that the questions we are studying here are different and possibly more fundamental than those studied in \cite{labericml}: is it possible, even
allowing exponential construction time, to build a decision tree whose expected cost is very close to the best possible expected cost achievable
and whose worst testing cost is very close to the best possible worst case achievable? How close can we get or better what is the best trade off
we can simultaneously guarantee?

For the prefix code problem
there are some studies related to the  simultaneous minimization
of  the expected testing cost and the worst case testing cost \cite{garey,larmore,fastalg,laber}.
The problem of constructing a prefix code 
is a particular case of the DFEP in which each object belongs to a distinct class, the testing costs are uniform 
and the set of tests is in one to one correspondence with the set of all binary
strings of length $n$ so that the test corresponding to a binary
string $b$ outputs 0 (1) for object $s_i$ if and only if the
$i^{th}$ bit of $b$ is 0 (1).

A number of algorithms with different  time complexities
were proposed to 
construct  decision trees with minimum expected path length (expected testing cost in DFEP terminology)
among the decision trees with depth (worst testing cost) at most $L$,
where $L$ is a given integer \cite{garey,larmore,fastalg}.

The results of Milidiu and Laber \cite{laber} imply
that for any instance  $I$ of the the prefix code problem,
 there is a decision tree $D$ such that for any integer  $c$,
with $ 0 < c \leq (n-1)- \lceil \log n \rceil $,
$Cost_W(D) - OPT_W(I) = c$ and
$Cost_E(D) - OPT_E(I) \leq 1/ \psi^{c-1}$,
where $\psi$ is the golden ratio $(1+\sqrt{5})/2$.



When the goal is to minimize only one 
measure (worst or expected testing cost), there are several 
algorithms in the literature to solve the particular version of the $DFEP$ in which 
each object belongs to a distinct class 
(\cite{garey_id,Kos_id,pandit_id,adler_id,guillory_id,laber_id,guillory2,gupta}).
 Approximation algorithms for the general version of the problem, where the number of classes  can be smaller than the number of objects, were presented by
   \cite{bellala}, \cite{golovin} and \cite{labericml}. 
For the minimization of the worst testing cost of DFEP, Moshkov  has studied the problem 
in the general case of multiway tests and non-uniform costs and provided 
an $O(\log n)$-approximation in \cite{Moshkov2}. 
Our algorithm in Section 3, generalizes Moshkov's algorithm to the value-dependent-test-cost variant of the DFEP
Moshkov \cite{Moshkov2} also proved  that
that 
no $o(\log n)$-approximation algorithm is possible under the standard  complexity assumption
$NP \not \subseteq DTIME(n^{O(\log \log n)}).$ The minimization of the worst testing cost is also investigated in 
\cite{conf/icml/GuilloryB11}  under the framework of covering and learning.
	Both \cite{bellala} and \cite{golovin} show 
 $O(\log (1/p_{min}))$ approximations for the expected testing cost (where $p_{min}$ is the minimum probability among the 
 objects in $S$) \---- the former for binary tests, and the latter for multiway 
 tests.


\section{Preliminaries and notation}
\label{sec:prelim}

In order to explain our results, we use $OPT_W(S,C,T,\mathbf{p},\mathbf{c})$ and $OPT_E(S,C,T,\mathbf{p},\mathbf{c})$, respectively, to denote the cost
of the decision tree with minimum worst testing cost
and  minimum expected testing cost for the input $(S,C,T,\mathbf{p},\mathbf{c})$.
Whenever the context permits (it will always permit) we  use the simpler notations
$OPT_W(S)$ and $OPT_E(S)$.

Let $(S,C,T,\mathbf{p},\mathbf{c})$ be an instance
of DFEP and let  $S'$ be a subset of $S$.
In addition, let $C'$ and $\mathbf{p}'$ be, respectively, the restrictions of $C$ and $\mathbf{p}$
to the set $S'$.
Our first observation is that every decision tree
$D$ for  $(S,C,T,\mathbf{p},\mathbf{c})$ is also a decision tree
for $(S',C',T,\mathbf{p}',\mathbf{c})$.
The following proposition is a direct consequence of this  observation.

\begin{proposition}
\label{prop:Subadditivity}
Let $(S,C,T,\mathbf{p},\mathbf{c})$ be an instance
of the DFEP and let  $S'$ be a subset of $S$.
Then, $OPT_E(S') \leq OPT_E(S)$ and
$OPT_W(S') \leq OPT_W(S)$.
\end{proposition}

We say that a pair of objects $(s_i,s_j)$  from a set  $S$
is {\em separable} if $s_i$ and $s_j$ belong to different classes.
For a set of objects $G$ we use 
$P(G)$ to denote the number of separable pairs in
$G$. In formulae,
\begin{equation}
P(G) = \sum_{i=1}^{k-1}\sum_{j=i+1}^k n_i n_j,
\end{equation}
where $n_i$ is the number of objects in $G$ that belong to class $i$.
We say that a test $t$ {\em separates} a pair of separable objects $(s,s')$
if $t(s) \neq t(s')$.

\section{A logarithmic approximation for value dependent testing costs}\label{sec:binary}

We first consider the goal of approximating optimal decision trees with respect to the 
worst testing cost. Recall that if we apply a test $t$ 
on an object $s \in S$, getting an answer $t(s)$, we pay a cost $c^{t(s)}(t)$. Thus, each test can be associated with 
$\ell$ different costs since $t(s) \in \{1,\ldots,\ell\}$. Note that now each 
branch of a decision tree is associated with a cost, while in the classical version of the problem each internal node is associated 
with a cost.

Our algorithm, called \textsc{DividePairs}, chooses the test ${t}$ that minimizes:

\begin{equation}\label{criterio}
\max_{1 \leq i \leq \ell} \left\{\frac{c^i(t)}{P(S) - P(S^i_t)}\right\}
\end{equation}
over all available tests for the root of the tree. Then the objects in $S$ are splitted according to the values of $t$ for each object, 
and \textsc{DividePairs} is recursively called for each (non empty) new group of objects. When all objects in a group are from the same class, a leaf is created. We analyze the approximation of the algorithm when  $\ell = 2$. 
Recall that we use $S^i_{t}$ to denote  
 the subset of objects of $S$ for which test $t \in T$ outputs
$i$.

In this case, each test $t \in T$ splits $S$ in two subsets: $S^1_t$ and $S^2_t$.

In order to analyze the algorithm, 
we use $Cost_W(S)$ to denote the cost of the decision tree that \textsc{DividePairs} constructs
for a set of objects $S$.
Let $\tau$ be the first test selected by \textsc{DividePairs}.
  We can write
 the ratio between the worst testing cost of the decision tree generated by
 \textsc{DividePairs} and the cost of the decision tree with minimum worst testing cost as

\begin{equation}\label{eqratio}
\frac{Cost_W(S)}{OPT_W(S)} = \frac{\max\{ c^1(\tau) + Cost(S^1_{\tau}),c^2(\tau) + Cost(S^2_{\tau})  \}}{OPT_W(S)}
\end{equation}

Let $q$ be such that $c^{q}(\tau) + Cost(S^{q}_{\tau}) = 
\max \{c^{1}(\tau) + Cost(S^{1}_{\tau}),c^{2}(\tau) + Cost(S^{2}_{\tau}) \}$ in equation 
(\ref{eqratio}). We have that:

\begin{equation}\label{eqratio2}
\frac{Cost_W(S)}{OPT_W(S)} = \frac{c^{q}(\tau) + Cost(S^{q}_{\tau})}{OPT_W(S)} \leq \frac{c^{q}(\tau)}{OPT_W(S)} + \frac{Cost(S^{q}_{\tau})}{OPT_W(S^{q}_{\tau})}
\end{equation}
where the inequality follows from Proposition \ref{prop:Subadditivity}. The 
following lemma shows  that $OPT_W(S)$ is at least $c^{q}(\tau) P(S)/(P(S) - P(S_{\tau}^{q}))$.

\begin{lemma}
$c^{q}(\tau) P(S)/(P(S) - P(S_{\tau}^{q}))$ is a lower bound on the worst testing cost of the optimal tree.
\end{lemma}

\textit{\textbf{Proof:}} First, we note 
that in the set of decision trees with optimal worst testing cost, 
there is a tree $D^*$ in which every internal node has two children. 
Let $v$ be an arbitrarily chosen internal node in $D^*$,
let $\gamma$ be the test associated with $v$ and let
$R \subseteq S$ be the set of objects associated with the leaves of the subtree rooted at $v$.
Let $i$ be such that $c^i(\tau)/(P(S) - P(S^i_{\tau}))$ is maximized and $j$ be such that $c^j(\gamma)/(P(S) - P(S^j_{\gamma}))$ 
is maximized.
We have that:
	
\begin{align}
\frac{c^{q}(\tau)}{P(S) - P(S_{\tau}^{q})}  
\leq \frac{c^i(\tau)}{P(S) - P(S_{{\tau}}^i)} 
 \leq \frac{c^j(\gamma)}{P(S) - P(S_{{\gamma}}^j)} \label{eqgreedy1}
 \\ \leq  \frac{c^j(\gamma)}{P(R) - P(R_{\gamma}^j)} \label{eqgreedy2} 
\end{align}

The last inequality in (\ref{eqgreedy1}) holds due 
to the greedy choice. To prove inequality (\ref{eqgreedy2}),
we only have to show that $P(S) - P(S_{\gamma}^j) \geq P(R) - P(R_{\gamma}^j)$.  
Let $r_{\gamma}^R$ (resp. $r_{\gamma}^S$) be  
the number of  pairs 
in $R$ (resp.\ $S$) separated by test $\gamma$.
Since  $R \subseteq S$ we have that $ r_{\gamma}^R \leq r_{\gamma}^S$ and
$P(R_{\gamma}^{i}) \leq P(S_{\gamma}^{i})$ for $i=1,2$.
Also, note that:

\begin{equation}\label{pairsS}
P(S) = r_{\gamma}^S + P(S_{\gamma}^{1}) + P(S_{\gamma}^{2})
\end{equation}

\begin{equation}\label{pairsR}
P(R) = r_{\gamma}^R + P(R_{\gamma}^{1}) + P(R_{\gamma}^{2})
\end{equation}

Hence, we have that $P(S) - P(S_{\gamma}^j) \geq P(R) - P(R_{\gamma}^j)$. Thus, we have concluded that 
inequality (\ref{eqgreedy2}) holds. 
	
For a node $v$, let $S(v)$ be the set of objects
associated with the leaves of the subtree rooted at $v$. 
Let $v_1, v_2, \dots, v_p$ be a root-to-leaf path on $D^*$ as follows:
$v_1$ is the root of the tree, and for each $i=1, \dots, p-1$ the node $v_{i+1}$ is a child of 
 $v_i$ associated with the branch $j$ that maximizes $c^j(t_i)/(P(S) - P(S^j_{t_i}))$, where $t_i$ is the test associated with $v_i$.  
 We denote by $c^{*}_{t_i}$ the cost that we have to pay going from $v_i$ to $v_{i+1}$.
It follows from inequaltity (\ref{eqgreedy2})
that  
\begin{equation}\label{eqaux}
 \frac{ \left [ P(S(v_i)) - P(S(v_{i+1}))\right ]c^{q}(\tau)}{P(S) - P(S_{{\tau}}^{q})} \leq c^{*}_{t_i}
\end{equation}
for $i=1,\ldots,p-1$.
Since the cost of the  path from $v_1$ to $v_p$ is not larger than the worst testing cost of
the optimal decision tree,
we have that

$$OPT_W(S) \geq \sum\limits_{i=1}^{p-1} {c^{*}_{t_i}} \geq \frac{c^{q}(\tau)}{P(S) - P(S_{\tau}^{q})} \sum\limits_{i=1}^{p-1} \left (
P(S(v_i))-P(S(v_{i+1})) \right) =   \frac{c^{q}(\tau)P(S)}{P(S) - P(S_{\tau}^{q})} 
,$$
where the second inequality follows from (\ref{eqaux})
and the last identity holds because $S(v_1)=S$ and $P(S(v_p))=0$.
$\qed$

Replacing the bound on 
$OPT_W(S)$
given by the previous lemma in  equation (\ref{eqratio2}) we get that

\begin{align}\label{eq:ratio_lb}
\frac{Cost_W(S)}{OPT_W(S)} \leq \frac{P(S)-P(S_{\tau}^{q})}{P(S)} +  \frac{Cost_W(S_{\tau}^{q})}{OPT_W(S_{\tau}^{q})}
\end{align}

Note that:

\begin{equation}\label{eq:firstterm}
\frac{P(S)-P(S_{\tau}^{q})}{P(S)} = \sum\limits_{i=1}^{P(S)-P(S_{\tau}^{q})}\Bigg(\frac{1}{P(S)}\Bigg) \leq \sum\limits_{i=1}^{P(S)-P(S_{\tau}^{q})}\Bigg(\frac{1}{P(S_{\tau}^{q}) + i}\Bigg)
\end{equation}

By induction on the number of  pairs,  we assume  that for each $G \subset S$, $Cost_W(G)/OPT_W(G) \leq H(P(G))$, where 
$H(n) = \displaystyle\sum\limits_{i=1}^{n} 1/i$. 
From (\ref{eq:ratio_lb}) and (\ref{eq:firstterm}) we have that

$$\frac{Cost_W(S)}{OPT_W(S)} \leq \sum\limits_{i=1}^{P(S)-P(S_{\tau}^{q})}\Bigg(\frac{1}{P(S_{\tau}^{q}) + i}\Bigg) + H(P(S_{\tau}^{q})) = H(P(S)) \leq 2\ln(n). $$
Thus,
we have the following theorem

\begin{theorem}
There is an $O(\log n)$ approximation for version of the DFEP with binary tests and  value dependent
costs.
\end{theorem}

\section{A bicriteria approximation}
\label{sec:bicriteria}
In this section,
we present an algorithm that provides a simultaneous approximation
for the minimization of expected testing cost and worst testing cost.
 There are examples in which the minimization of the expected testing
cost produces a decision tree with high worst testing cost, and the 
minimization of the worst testing
cost produces a decision tree with high expected testing cost \cite{labericml}.
Therefore, it makes sense to look for a trade-off between minimizing
both measures.

Given a positive number $\rho$, two decision trees $D_E$ and $D_W$ for the instance $(S,C,T,\mathbf{p},\mathbf{c})$,
 the former  with expected testing cost  $E$ and the latter with worst testing cost  $W$, we devise a polynomial time procedure 
 to construct a new decision tree $D$, from $D_E$ and $D_W$, with
 expected cost at most $(1 + 1/ \rho) E$ and
  worst testing cost at most $(1+ \rho) W$.
The procedure is very simple:

\medskip

{\tt CombineTrees}($D_E$,$D_W$,$\rho$)

\begin{enumerate}

\item Define a node $v$ from $D_E$ as replaceable
if the cost of the path from the root of $D_E$ to  $v$ (including $v$) is
at least $\rho W$ and the cost of the path from the root of $D_E$ to the parent of $v$
is smaller than $\rho W$. At this step we traverse $D_E$ to 
find the set  $R$ of the replaceable nodes.

\item For every node $v \in R$ do 

\begin{enumerate}

\item Let  $S(v)$  be  the set of objects associated with leaves 
located at the subtree rooted at $v$ in $D_E$. In addition, let 
$D^{S(v)}_W$ be a decision tree for $S(v)$ obtained
by disassociating every object in $S-S(v)$ from $D_W$.

\item Replace the subtree of $D_E$ rooted at $v$ with the decision tree
$D^{S(v)}_W$

\end{enumerate}

\item Return the tree  $D$ obtained by the end of Step 2. 

\end{enumerate}

\normalsize

\begin{theorem}
The decision tree $D$ has expected testing cost at most $(1+1/ \rho)E$ and worst testing cost at most $(1+\rho)W$.
\end{theorem}
\begin{proof}
First we argue that the worst testing cost of $D$ is at most
$(1+\rho)W$.
Let $s$ be an object in $S$.
If $s$ is not a descendant of a replaceable node in $D_E$ then
the cost of the path from the root of $D_E$ to $s$ is at most $\rho W$.
Since this path remains the same in  $D$,
we have that the cost to reach $s$ in $D$ is at most $\rho W$.
On the other hand, if $s$ is  a descendant of a replaceable node $v$ in $D_E$, then
the cost to reach $s$ in $D$ is at most $(1+\rho)W$ because the cost of the path from the root of $D$  to the parent of $v$ is at 
 most $\rho W$ and the cost to reach $s$ from the root of the tree $D_W^{S(v)}$ is at most
$W$.

Now, we prove that the expected testing
 cost of $D$ is at most
$(1 + 1/ \rho) E$.
For that it is enough 
to show that for every object $s \in S$, the cost
to reach $s$ in $D$ is at most $(1 + 1/ \rho)$ times the cost of
reaching $s$ in $D_E$.
We split the analysis into two cases:

{\bf Case 1.}
 $s$ is not a descendant of a replaceable node in $D_E$. In this case, the cost
to reach $s$ in  $D_E$ is equal to the cost of reaching $s$ in $D$.

\medskip
{\bf Case 2}.  $s$ is  a descendant of a replaceable node $v$ in $D_E$.
Let $K$ be the cost of the path from the root of $D_E$ to  $v$. 
Then, the cost  to reach $s$ in $D_E$ 
is at least $K$.   In addition, since $v$ is replaceable we have that
 $K \geq \rho W$.
On the other hand, the cost to reach $s$ in 
$D$ is at most $\rho W +W$. 
Since $K \geq \rho W$ we have that the cost
to reach $s$ in $D$ is at most $(1+1/\rho)$ times the cost of
reaching $s$ in $D_E$.
\end{proof}

We can improve the approximation for 
the case where the costs are uniform.
In this case, we can assume unitary testing costs so that  $W$ is the height
of the decision tree $D_W$.
Let  $L$ and $M$, with $L<M$, be two positive integers whose values  will be defined during our analysis.

To obtain a better approximation,  we consider an algorithm that picks the decision tree, say  $D$, with minimum expected testing cost among
the decision trees $D^{L},D^{L+1},\ldots,D^{M}$, where $D^i$ is the decision tree returned
by {\tt CombineTrees} when it is executed
with parameters $(D_E,D_W, i/W)$.
It follows from the previous theorem that 
$$Cost_W(D) \leq (1+M/W) W =M+ W.$$

The analysis of the  expected testing cost of $D$ is more involved.
First,  we have that 

\begin{equation}
\label{eq:ub-uniformcosts} 
Cost_E(D) = \min_{i=L,L+1\ldots,M} \{Cost_E(D^i)\} \leq \frac{ \sum_{i=L}^M Cost_E(D^i)}{M-L+1} 
\end{equation}

Let $H$ be the height
of the decision tree $D_E$.
For $j=1,\ldots,H$, let 
  $C_j$ be the contribution of the leaves located at  level $j$ for
  the cost of $D_E$ so that $Cost_E(D_E)=\sum_{j=1}^H C_j$.
  It follows that 
$$ Cost_E(D^i) \leq \sum_{j=1}^i C_j + \sum_{j=i+1}^H \frac{C_j (i+W)}{j},$$ 
because the objects associated with leaves that are located at levels smaller than
or equal to $i$ are not modified from $D_E$ to $D^i$  while the remaining objects are located at levels
smaller than or equal to $i+W$ in $D^i$. Note that $C_j/j$ in the previous inequality is the sum of the probabilities
of the leaves at level $j$.
By replacing the last
expression in (\ref{eq:ub-uniformcosts}) and
grouping the terms around the $C_j$'s we get that

$$\frac{Cost_E({D})}{Cost_E(D_E)} \leq \frac{ \sum_{j=1}^H \alpha_j C_j }{ \sum_{j=1}^H  C_j} \leq   max_j\{ \alpha_j\},$$
where

\[ \alpha_j = \left\{ \begin{array}{ll}
         1 & \mbox{if $j \leq L$};\\ \\
         \frac{M-j +1 + \frac{(j-L)W+ (j-L)(j-1+L)/2 }{j}}{M-L+1} & \mbox{if $L < j \leq M$}  \\ \\
        \frac{W +(M+L)/2}{j} & \mbox{if $j \geq M+1$}
        .\end{array} 
                \right. \]

First, note that the maximum of $\alpha_j$ in the range $ j \geq M+1$ is
$(W +(M+L)/2)/j$, which is attained when $j=M+1$.
Moreover, if we replace $j=M+1$ in the formula 
of $\alpha_j$ for  the range $L <  j \leq M$ we get exactly 
$(W +(M+L)/2)/j$. Thus, 
it follows that 
$$ \frac{Cost_E({D})}{Cost_E(D_E)} \leq  \max_{j \in (0,\infty) } \left \{  \frac{M-j +1 + \frac{(j-L)W+ (j-L)(j-1+L)/2 }{j}}{M-L+1}  \right \} $$
By simple calculus we can conclude that the expression attains 
the maximum when $j= \sqrt{L^2-L+2LW}$.
Thus,

\begin{equation}
\label{eq:finalbound}
 \frac{Cost_E({D})}{Cost_E(D_E)}  \leq 1+ \frac{L+W - \sqrt{L^2-L+2LW} -1/2 }{M-L+1} \leq 1+ \frac{L+W -\sqrt{L^2+2LW}}{M-L+1}.
 \end{equation}
To verify the last  inequality we need
to do some calculations (squaring the terms) and use the fact that  $W,L\geq 1$.

Let $r$ be a number in the interval
$[0,1/W]$.
We can verify that
the righthand side of the  equation (\ref{eq:finalbound}) is upper bounded
by $ 1 +  2/(\rho^2 + 2\rho)$ 
whenever $M= \rho W$    and $L=  W  (t + r)$, where $t=\frac{\rho^2}{2\rho+2}$ (the proof is presented in the appendix). 

Thus, by setting  $M= \rho W$    and
$L= \lceil W  \rho^2 /  (2 \rho+2) \rceil$, where
$\rho$  is a positive number that can be written as $i/W$
for some integer $i$,
we obtain  the following theorem.

\begin{theorem}
Let ${\cal I}= (S,C,T,\mathbf{p},\mathbf{c})$
an instance of the DFEP where all the tests have unitary  costs.
Given  two decision trees $D_E$ and $D_W$ for the instance ${\cal I}$,
 the former  with expected testing cost  $E$ and the latter with height  $W$ and a positive number $\rho$ that
can be written as $i/W$ for some integer $i$, there exists
 a polynomial time algorithm that construts a  decision tree $D$ with  height at most $(1+ \rho) W $
 and  expected testing cost at most $\left (1 +  \frac{2}{\rho^2 + 2\rho}  \right) E.$ 
 \end{theorem}

As an  example, for $\rho=2$ this new algorithm guarantees that the 
expected testing cost is at most $(5/4)E$ while
the initial algorithm guarantees a $1.5E$ upper bound.

\section{Conclusions}

We presented a polynomial time procedure that given a parameter $\rho >0$, a decision tree 
$D_W$ with worst testing cost $W$ and a decision tree $D_E$ with expected testing cost $E$, produces
a decision tree $D$ with worst testing cost at most $(1+\rho)W$ and expected
testing cost at most $(1+1/\rho)E$. When the costs are uniform,
the bound can be improved to $(1+\rho)W$ and $(1+2/(\rho^2 +2 \rho))E$. 
The main question that remains open in this topic is whether for 
every $\epsilon >0$, there is some integer $n_0$ such that
every instance $I$ of the DFEP with more than
$n_0$ objects admits a tree $D$ such that $cost_E(D) \leq (1+\epsilon) OPT_E(I)$ and 
$cost_W(D) \leq (1+\epsilon) OPT_W(I)$.
For the prefix code problem, a particular version of the DFEP explained  in the introduction, this result
holds \cite{laber}.

We also presented an approximation algorithm for the 
extended version of the DFEP where the 
cost of the tests  depend also on the answers.
For the particular case where the tests are binary, our algorithm provides a logarithmic approximation which is the best
approximation unless ${\cal P}={\cal NP}$.
An interesting question that deserves more investigation is if there exists also a logarithmic approximation algorithm for the most general 
case where the  tests can output more than two values.

\bibliographystyle{abbrv}
\bibliography{esa2014}

\appendix

\section{Calculation of Section \ref{sec:bicriteria}}

Let $r$ be a number in the interval $[0,1/W]$. We have to prove that:

$$\left [ (t+r)W+W -W\sqrt{(t+r)^2+2(t+r)} \right] \leq \frac{2(\rho W-(t+r)W+1)}{(\rho^2+2\rho)}$$

By simple algebraic manipulations we conclude that we have to prove that

$$(\rho^2+2\rho)\left [ (t+r+1) -\sqrt{(t+r)^2+2(t+r)} \right] \leq 2(\rho-(t+r)+1/W),$$

or equivalently,

$$(\rho^2+2\rho)(t+r+1) - 2(\rho-(t+r)+1/W) \leq (\rho^2+2\rho)\sqrt{(t+r)^2+2(t+r)} $$ 

Replacing $t=\frac{\rho^2}{2\rho+2}$ and using the fact that $r \leq 1/W$, it suffices to show

$$(\rho^2+2\rho)\left (\frac{\rho^2}{2\rho+2}+r \right ) +\rho^2 +2\frac{\rho^2}{2\rho+2} \leq (\rho^2+2\rho)\sqrt{(t+r)^2+2(t+r)} $$

$$(\rho^2+2\rho)\left (\frac{\rho^2}{2\rho+2}+r \right) +\rho^2 +2\frac{\rho^2}{2\rho+2} \leq$$
$$ \frac{\rho^2+2\rho}{2\rho+2}\sqrt{\rho^4+2\rho^2(2\rho+2)r+(2\rho+2)^2r^2+
2(2\rho+2)\rho^2+2(2\rho+2)^2r} $$

$$(\rho^2+2\rho)\rho^2 +(\rho^2+2\rho)r(2\rho+2)+\rho^2(2\rho+2)+2\rho^2 \leq $$
$$ (\rho^2+2\rho)\sqrt{\rho^4+2\rho^2(2\rho+2)r+(2\rho+2)^2r^2+
2(2\rho+2)\rho^2+2(2\rho+2)^2r} $$

$$\rho^4+4\rho^3+4\rho^2+(2\rho^3+6\rho^2+4\rho)r
   \le  $$ 
 $$ (\rho^2+2\rho)\sqrt{\rho^4+4\rho^3+6\rho^2 +(2\rho+2)^2 r^2+
4(\rho+1)(\rho^2+2\rho+2)r} $$

This can be shown by verifying that the following inequalities hold:
$$(\rho^4+4\rho^3+4\rho^2)^2 \leq (\rho^2+2\rho)^2(\rho^4+4\rho^3+6\rho^2) $$

$$(2\rho^3+6\rho^2+4\rho)^2r^2 \leq (\rho^2+2\rho)^2 (4\rho^2+8\rho^2+4)r^2   $$

and
$$ 2 (\rho^4+4\rho^3+4\rho^2) (2\rho^3+6\rho^2+4\rho)r \leq (\rho^2+2\rho)^2 4(\rho+1)(\rho^2+2\rho+2)r$$

\end{document}